\documentclass[10pt]{article}
\usepackage{amsfonts,amsthm,amsmath,amssymb,mathrsfs,url, braket, graphicx}
\usepackage[usenames,dvipsnames]{color}

\title{Can One Detect Whether a Wave Function Has Collapsed?}
\author{
Charles Wesley Cowan\footnote{Department of Mathematics, Rutgers University, Hill Center,  
     110 Frelinghuysen Road, Piscataway, NJ 08854-8019, USA.}\ \footnote{E-mail:
     cwcowan@math.rutgers.edu}\ \ and 
Roderich Tumulka$^*$\footnote{E-mail: tumulka@math.rutgers.edu}
}
\date{December 30, 2013}

\newcommand{\Hilbert}{\mathscr{H}}
\newcommand{\E}{\mathcal{E}}
\newcommand{\EEE}{\mathbb{E}}
\newcommand{\PPP}{\mathbb{P}}
\newcommand{\RRR}{\mathbb{R}}
\newcommand{\CCC}{\mathbb{C}}
\newcommand{\SSS}{\mathbb{S}}

\newcommand{\NNN}{\mathbb{N}}
\newcommand{\scp}[2]{\langle #1|#2 \rangle}
\newcommand{\pr}[1]{| #1 \rangle \langle #1 |}
\newcommand{\be}{\begin{equation}}
\newcommand{\ee}{\end{equation}}

\newcommand{\vv}{\boldsymbol{v}}
\newcommand{\vw}{\boldsymbol{w}}
\newcommand{\vsigma}{\boldsymbol{\sigma}}
\DeclareMathOperator{\tr}{tr}
\DeclareMathOperator{\diag}{diag}
\theoremstyle{plain}
	\newtheorem{prop}{Proposition}
	\newtheorem{theorem}{Theorem}
	
	\newtheorem{corollary}{Corollary}

\begin{document}
\maketitle

\begin{abstract}
Consider a quantum system prepared in state  $\psi$, a unit vector in a $d$-dimensional Hilbert space. Let $b_1,\ldots,b_d$ be an orthonormal basis and suppose that, with some probability $0<p<1$, $\psi$ ``collapses,'' i.e., gets replaced by $b_k$ (possibly times a phase factor) with Born's probability $|\scp{b_k}{\psi}|^2$. The question we investigate is: How well can any quantum experiment on the system determine afterwards whether a collapse has occurred? The answer depends on how much is known about the initial vector $\psi$. We provide a number of different results addressing several variants of the question. In each case, no experiment can provide more than rather limited probabilistic information. In case $\psi$ is drawn randomly with uniform distribution over the unit sphere in Hilbert space, no experiment performs better than a blind guess without measurement; that is, no experiment provides any useful information.

\medskip

Key words: collapse of the wave function; limitations to knowledge; absolute uncertainty; empirically undecidable; quantum measurements; foundations of quantum mechanics; Ghirardi-Rimini-Weber (GRW) theory; random wave function.
\end{abstract}

\section{Introduction}\label{sec:introduction}

We consider a quantum system whose wave function may or may not have collapsed, and ask whether experiments on the system can provide us with information about whether it has collapsed, either in the case we know the system's initial wave function or in the case we do not. 

The main motivation for this question \cite{GTZ07,CT12b} comes from the Ghirardi--Rimini--Weber (GRW) theory \cite{GRW86,AGTZ06} of quantum mechanics, which solves the paradoxes of quantum mechanics by replacing the Schr\"odinger equation with a stochastic process in which wave functions sometimes collapse in a random way, also without the intervention of an ``observer.'' As we elucidate in detail elsewhere \cite{GTZ07,CT12b}, the results presented here imply that the inhabitants of a universe governed by the GRW theory cannot discover all facts true of their universe---there are limitations to their knowledge. Specifically, they cannot measure the number of collapses in a given physical system during a given time interval, although this number is well defined; in fact, as we show, they cannot reliably find out whether any collapse at all has occurred in the system.

However, the questions that we investigate in this paper can also be considered in the framework of orthodox quantum mechanics and are, in our opinion, of interest in their own right. The basic type of question is as follows.

Consider a quantum system $S$ with Hilbert space $\Hilbert$ of finite dimension $d \in \NNN$, $d \geq 2$. Let
\begin{equation}\label{eqn:SSSdef}
\SSS = \{\psi \in \Hilbert: \|\psi\|=1\}
\end{equation}
denote the unit sphere in $\Hilbert$, and let $B = \{b_1, \ldots, b_d\}$ be an orthonormal basis of $\Hilbert$. Suppose that the ``initial'' wave function of $S$ was $\psi\in\SSS$ but with probability $p$ a collapse relative to $B$ has occurred. That is, suppose that the wave function of $S$ is the $\SSS$-valued random variable $\psi'$ defined to be
\begin{equation}\label{eqn:psi-prime-def}
\psi' = \begin{cases}
\psi & \text{with probability }1-p\\
\frac{\scp{b_k}{\psi}}{|\scp{b_k}{\psi}|}b_k & \text{with probability }p\,\bigl| \scp{b_k}{\psi} \bigr|^2\text{ for }k=1,\ldots,d\,.
\end{cases}
\end{equation}
Is there an experiment on $S$ that would reveal whether a collapse has occurred? Or at least provide probabilistic information about whether a collapse has occurred? What is the best experiment to obtain such information? We take $p$ and $B$ to be known;\footnote{Actually, the problem depends on $B$ only through its equivalence class, with the basis $\{e^{i\theta_1}b_1,\ldots,e^{i\theta_d}b_d\}$ regarded as equivalent to $B$ for arbitrary $\theta_1,\ldots,\theta_d\in\RRR$. So we take the equivalence class of $B$ (or, equivalently, the collection of $d$ 1-dimensional subspaces $\CCC b_k$) to be known; nevertheless, we often find it convenient to speak as if $B$ were given.} $\psi$ may or may not be known.

We may imagine the following story. Alice prepares $S$ with wave function $\psi \in \SSS$. The Hamiltonian of $S$ is 0. Alice leaves the room briefly. In her absence, with probability $0 < p < 1$, Bob enters the room. Bob performs on $S$ a quantum measurement of an observable with eigenbasis $B$, causing the wave function to collapse. Bob then sneaks back out. When Alice returns to the room, she wishes to know whether Bob has been there and tampered with her system.\footnote{In GRW theory, spontaneous collapses may replace Bob's intervention. If Alice lets $S$ sit for a while $t$ with zero Hamiltonian, the wave function collapses spontaneously, essentially relative to the position basis, with probability $p=1-\exp(N\lambda t)$, where $N$ is the number of particles in $S$ and $\lambda$ is a constant of nature that is in principle measurable; see \cite{CT12b} for more detail.} In short, she wants to determine whether or not $S$ has collapsed from its original state. To this end, Alice would like to perform an experiment on $S$. The difficulty Alice is faced with is the well-known problem of distinguishing between two non-orthogonal states---collapsed and non-collapsed. As the system cannot collapse to a state orthogonal to $\psi$, this difficulty is unavoidable. What Alice \textit{is} able to determine depends a great deal on what she knows about the initial state of $S$. We distinguish the following situations:

\begin{itemize}
\item[(i)] \textbf{Complete Information:} Alice knows the initial vector $\psi$.
\item[(ii)] \textbf{Partial Information:} Alice does not know $\psi$, but knows $\psi$ was sampled from $\SSS$ with known distribution $\mu$.
\item[(iii)] \textbf{No Information:} Alice knows nothing about $\psi$.
\end{itemize}

Note that (i) is in fact a special case of (ii), as a specific $\psi$ may be given via a delta distribution on $\SSS$. Nevertheless, it is a case worth distinguishing as in it we can present much stronger results. In this paper, we discuss (i) and (ii) in detail. Mathematically, case (iii) is of a very different flavor to (i) and (ii). Therefore, we discuss (iii) elsewhere \cite{CT12c} and report here only the main results.

The rest of the paper is organized as follows. In Sec.~\ref{sec:povms}, we set up the POVM that mathematically represents Alice's experiment. In Sec.~\ref{sec:complete}, we discuss our problem in the case that $\psi$ is known. In Sec.~\ref{sec:incomplete}, we discuss the case that $\psi$ is unknown but random with known distribution $\mu$. In Sec.~\ref{sec:no-information-experimentation}, we give a summary of our results in \cite{CT12c} about what is possible when Alice has no information about $\psi$.


\section{Mathematical Tools}

As a preparation, we describe some key facts and concepts that we will use.

\subsection{POVMs: A Mathematical Description of Experiments}
\label{sec:povms}

An experiment $\E$ is carried out on a system $S$ and yields a (usually random) outcome $Z$ in some value space $\mathcal{Z}$. A relevant fact for the mathematical treatment of our question is this: For every conceivable experiment $\E$ that can be carried out on $S$, there is a positive-operator-valued measure (POVM) $M(\cdot)$ on $\mathcal{Z}$ acting on $\Hilbert$ such that the probability distribution of $Z$, when $\E$ is carried out on $S$ with wave function $\psi\in\SSS$, is given by
\begin{equation}\label{eqn:POVM-probability-statement}
\PPP(Z\in\Delta) = \scp{\psi}{M(\Delta)|\psi}
\end{equation}
for all measurable sets $\Delta\subseteq\mathcal{Z}$. 

The statement containing \eqref{eqn:POVM-probability-statement} was proved for GRW theory in \cite{GTZ07} and for Bohmian mechanics in \cite{DGZ04}. In orthodox quantum mechanics, the theorem is true as well, taking for granted that, after $\E$, a quantum measurement of the position observable of the pointer of $\E$'s apparatus will yield the result of $\E$.

It is important to note that while every experiment $\E$ can be characterized in terms of a POVM, it is not necessarily true that every POVM is associated with a realizable experiment.  For our purposes, so as to answer the question ``Did collapse occur?'', it suffices to consider \emph{yes-no-experiments}, i.e., those with $\mathcal{Z}=\{\text{yes},\text{no}\}$; for them the POVM $M(\cdot)$ is determined by the operator
\begin{equation}\label{eqn:povm-yes-operator}
E=M(\{\text{yes}\})\,.
\end{equation}
$I-E$ is the operator corresponding to $\text{no}$, $M(\{\text{no}\})$. By the definition of a POVM, $E$ must be a positive\footnote{We take the word ``positive'' for an operator to mean $\braket{\psi|E|\psi}\geq 0$ for all $\psi\in\Hilbert$, equivalently to its matrix (relative to any orthonormal basis) being positive semi-definite; we  denote this by $E\geq 0$.} operator such that $I-E$ is positive too; it is otherwise arbitrary. Thus, \emph{we can characterize every possible yes-no experiment $\E$ mathematically by a self-adjoint operator $E$ with spectrum in $[0,1]$, $0 \leq E \leq I$.} As noted, this is a larger set than the class of ``realizable'' experiments, but by proving results over the set of POVMs (in this case of yes-no experiments, proving results over the set of self-adjoint operators with the appropriate spectrum), the results necessarily cover all possible realizable experiments.

\subsection{Reliability}

We define the \emph{reliability} of a yes-no experiment to be the probability that its outcome correctly answers our question---in this case, the probability that the experiment correctly determines whether collapse has occurred. We use this quantity as a measure for how well an experiment performs for our purpose. In a scenario in which the initial wave function $\psi$ is known (as well as the a-priori probability $p$ of collapse), the reliability of an experiment $\E$ with outcome $Z\in\{\text{yes},\text{no}\}$ is
\begin{equation}\label{eqn:psi-reliability-def}
R_{\psi, p}(\E) = \PPP( Z = \text{yes}, \text{collapse} ) + \PPP( Z = \text{no}, \text{no collapse}).
\end{equation}
The most basic result of this paper (Thm.~\ref{thm:complete-information-imperfect-reliability} below) asserts the impossibility of detecting a collapse with perfect reliability, i.e., $R_{\psi,p}(\E)<1$ for all experiments $\E$, all $\psi\in\SSS$, and all $0<p<1$.

\subsection{Helstrom's Theorem}

We may embed our problem of detecting collapse in a larger class of problems, that of distinguishing between two density matrices $\rho_1\neq \rho_2$. Consider the following story: Bob gives to Alice a system $S$; with probability $p$, he has prepared $S$ to have density matrix $\rho_1$, and with probability $1-p$, he has prepared $S$ to have density matrix $\rho_2$. Alice would like to perform an experiment on $S$ to determine, at least with high probability, which of the two density matrices was used (in this particular individual case). 

The problem of detecting whether collapse has occurred is included as a special case: If $\psi$ is known, then $\rho_2=\pr{\psi}$ and $\rho_1 = \diag \pr{\psi}$, where ``$\diag$'' is the diagonal part of an operator relative to the basis $\{b_1,\ldots,b_d\}$,
\be\label{diagdef}
\diag E = \sum_{k = 1}^d \ket{b_k}\braket{b_k|E|b_k} \bra{b_k}
\ee
for any operator $E$. If $\psi$ is random with known distribution $\mu$, then $\rho_2$ is the density matrix corresponding to $\mu$, and $\rho_1 = \diag \rho_2$.

Let $\E$ be Alice's experiment to be performed on $S$, with two possible outcomes: If $Z = 1$ then Alice guesses the density matrix was $\rho_1$, if $Z=2$ then $\rho_2$. The POVM associated with $\E$ consists of the operators $0\leq E_1\leq I$ and $E_2=I-E_1$.
We again define the reliability as the probability that the outcome of the experiment correctly retrodicts which density matrix was used. We find that it is
\begin{equation}\label{RHelstrom}
\begin{split}
R_{\rho_1, \rho_2, p}(\E) 
& = p \PPP( Z = 1 | \rho_1 ) + (1-p) \PPP( Z = 2 | \rho_2 ) \\
& = p \tr \left[ \rho_1 E_1 \right] + (1-p) \tr \left[ \rho_2 (I - E_1) \right ] \\
& = 1-p + \tr \left[ A E_1 \right]
\end{split}
\end{equation}
with
\begin{equation}\label{A-def-general}
A = p \rho_1 - (1-p) \rho_2\,.
\end{equation}
In particular, the reliability depends on $\E$ only through the operator $E_1$; that is, different experiments with equal $E_1$ have equal reliability. For this reason, we will, when convenient, write $R_{\rho_1, \rho_2, p}(E_1)$ 
instead of $R_{\rho_1, \rho_2, p}(\E)$.

The optimal $E_1$ and its reliability
\begin{equation}
R^{\max}_p(\rho_1, \rho_2) = \max_{0 \leq E_1 \leq I} R_{\rho_1, \rho_2, p}(E_1)
\end{equation}
can be characterized as follows.

\begin{theorem}[Helstrom \cite{Hel76}]\label{thm:rho1rho2}
For $0 \leq p \leq 1$ and any density matrices $\rho_1, \rho_2$,
\begin{equation}\label{general-optimality}
R^{\max}_p(\rho_1, \rho_2) = (1-p) + \lambda^+ = p - \lambda^-,
\end{equation}
where $\lambda^+\geq 0$ and $\lambda^-\leq 0$ are, respectively, the sum of the positive eigenvalues (with multiplicities) and that of the negative eigenvalues of $A$ as in \eqref{A-def-general}.
The optimal operators $E_1=E_\mathrm{opt}$ for which this maximum is attained, $R^{\max}_p(\rho_1, \rho_2) = R_{\rho_1, \rho_2, p}(E_\mathrm{opt})$, are those satisfying
\be\label{Eoptrho12}
P^+_A \leq E_\mathrm{opt} \leq P^+_A + P^0_A\,,
\ee
where $P^+_A$ is the projection onto the positive spectral subspace of $A$, i.e, onto the sum of all eigenspaces of $A$ with positive eigenvalues, and $P^0_A$ is the projection onto the kernel of $A$.
\end{theorem}


\section{Complete Information}
\label{sec:complete}

In this section, we operate under the assumption that Alice knows $\psi$ precisely, and thus has complete information about the initial state of $S$. Let $\E$ be a yes-no experiment and $E$ the operator associated with the outcome ``yes.'' The reliability is found, for example from \eqref{RHelstrom} using $\tr(X\diag Y) = \tr(Y \diag X)$, to be
\begin{equation}\label{eqn:reliability-psi-formula}
R_{\psi, p}(\E)= R_{\psi,p}(E) = p \braket{\psi | \diag E | \psi } + (1-p) \braket{\psi | I - E | \psi}\,.
\end{equation}

\subsection{Perfect Reliability Is Impossible}

\begin{theorem}\label{thm:complete-information-imperfect-reliability}
$R_{\psi, p}(E) < 1$ for all operators $0\leq E\leq I$, all $0<p<1$, and all $\psi\in \SSS$. That is, for $0<p<1$ and known $\psi$, there is no yes-no-experiment $\E$ that can correctly determine with probability $1$ whether or not $S$ has collapsed. 
\end{theorem}

\begin{proof}
Without loss of generality, we may take $\braket{b_k | \psi} \neq 0$ for all $k = 1, \ldots, d$. If this did not hold for some $b_k$, that $b_k$ lies orthogonally to the initial state of the system. As such, collapse to $b_k$ occurs with probability $0$, and such an event may be excluded from consideration. The subspace generated by $b_k$ may in that case be ignored, and the problem treated in a smaller dimension. In the extreme event that only one $\braket{b_k|\psi}$ is nonzero, a collapse will leave $\psi$ unchanged, so it is obviously impossible to determine whether collapse has occurred; in fact, $R_{\psi=b_k,p}(E) \leq \max(p,1-p)<1$.

Assume now $\braket{b_k|\psi}\neq 0$.
The probability of $\E$ giving a false negative is 
\begin{equation}
\begin{split}
\PPP(\text{false negative}) & = \PPP( Z = \text{no}, \text{collapse}) \\
& = \PPP(\text{collapse}) \PPP(Z = \text{no}| \text{collapse}) \\
& = p \bigl(1 - \PPP(\text{yes} | \text{collapse})\bigr) \\
& = p \bigl(1 - \braket{\psi|\diag E|\psi} \bigr) \\
& = p \Bigl(1 -  \sum_{k = 1}^d \braket{b_k | E |b_k} \bigl| \scp{b_k}{\psi} \bigr|^2 \Bigr).\\
\end{split}
\end{equation}
Given that $p > 0$, a false negative rate of 0 requires that $\sum_{k = 1}^d \braket{b_k | E |b_k} \bigl| \scp{b_k}{\psi} \bigr|^2  = 1$. However, $0 \leq \braket{b_k | E |b_k} \leq 1$ for each $k$. Since $\sum_{k = 1}^d \bigl| \scp{b_k}{\psi} \bigr|^2  = \|\psi\|^2 = 1$, a false negative rate of 0 requires $\braket{b_k | E |b_k} = 1$ for each $k$.
This in turn forces $\tr E = d$. Since the eigenvalues of $E$ are restricted to $[0,1]$, all eigenvalues of $E$ must be $1$, hence $E$ must be the identity $I$. However, $E = I$ gives a false positive probability of
\begin{equation}
\begin{split}
\PPP(\text{false positive}) & = \PPP(Z = \text{yes}, \text{no collapse}) \\
& = \PPP(\text{no collapse}) \PPP( Z = \text{yes} | \text{no collapse}) \\
& = (1-p) \braket{ \psi | E | \psi } \\
& = (1-p) \braket{ \psi | I | \psi } \\
& = (1-p) \|\psi\|^2 \\
& = 1-p\,. 
\end{split}
\end{equation}
Therefore, a false negative rate of $0$ forces a false positive rate of $1 - p > 0$. The probability of an incorrect outcome can never be made $0$, unless collapse is guaranteed or forbidden. 
\end{proof}

In Thm.~\ref{thm:complete-information-imperfect-reliability}, allowing experiments with more than two outcomes obviously does not improve the situation.

Defining
\begin{equation}\label{eqn:r-max-psi}
R^{\max}_p(\psi) = \max_{0 \leq E \leq I} R_{\psi, p}(E),
\end{equation}
Thm.~\ref{thm:complete-information-imperfect-reliability} means that $R^{\max}_p(\psi) < 1$.


\subsection{Blind Guessing: The Trivial Experiment}\label{sec:blind-guessing}

We consider the following ``trivial'' experiments: Independently of what Alice actually knows about the initial state of $S$, she declares that collapse has occurred. This corresponds to taking $E = I$. Alternately, independently of what Alice knows about the initial state of $S$, she declares that collapse has not occurred. This corresponds to taking $E = 0$. Since collapse occurs with probability $p$, declaring collapse has occurred every time will be correct with probability $p$. Similarly, declaring collapse never occurs will be correct with probability $1-p$.

We may combine these approaches into a single experiment, $\E_\emptyset$, which we will refer to as blind guessing. If collapse is more probable, declare collapse. Else, declare no collapse. We define $\E_\emptyset$ such that
\begin{equation}\label{eqn:blind-guessing-operator}
E_\emptyset = \begin{cases}
0 & \text{if } p \leq 1/2 \\
I & \text{if } p > 1/2 \end{cases}.
\end{equation}
Since $R_{\psi, p}(\E_\emptyset) = p$ for $p \leq 1/2$ and $R_{\psi, p}(\E_\emptyset) = 1-p$ for $p > 1/2$, we have that 
\begin{equation}
R_{\psi, p}(\E_\emptyset) = \max(p, 1-p).
\end{equation}
This function of $p$ is depicted in Fig.~\ref{fig4}.

\begin{figure}[h]
\begin{center}
\includegraphics[width=5cm]{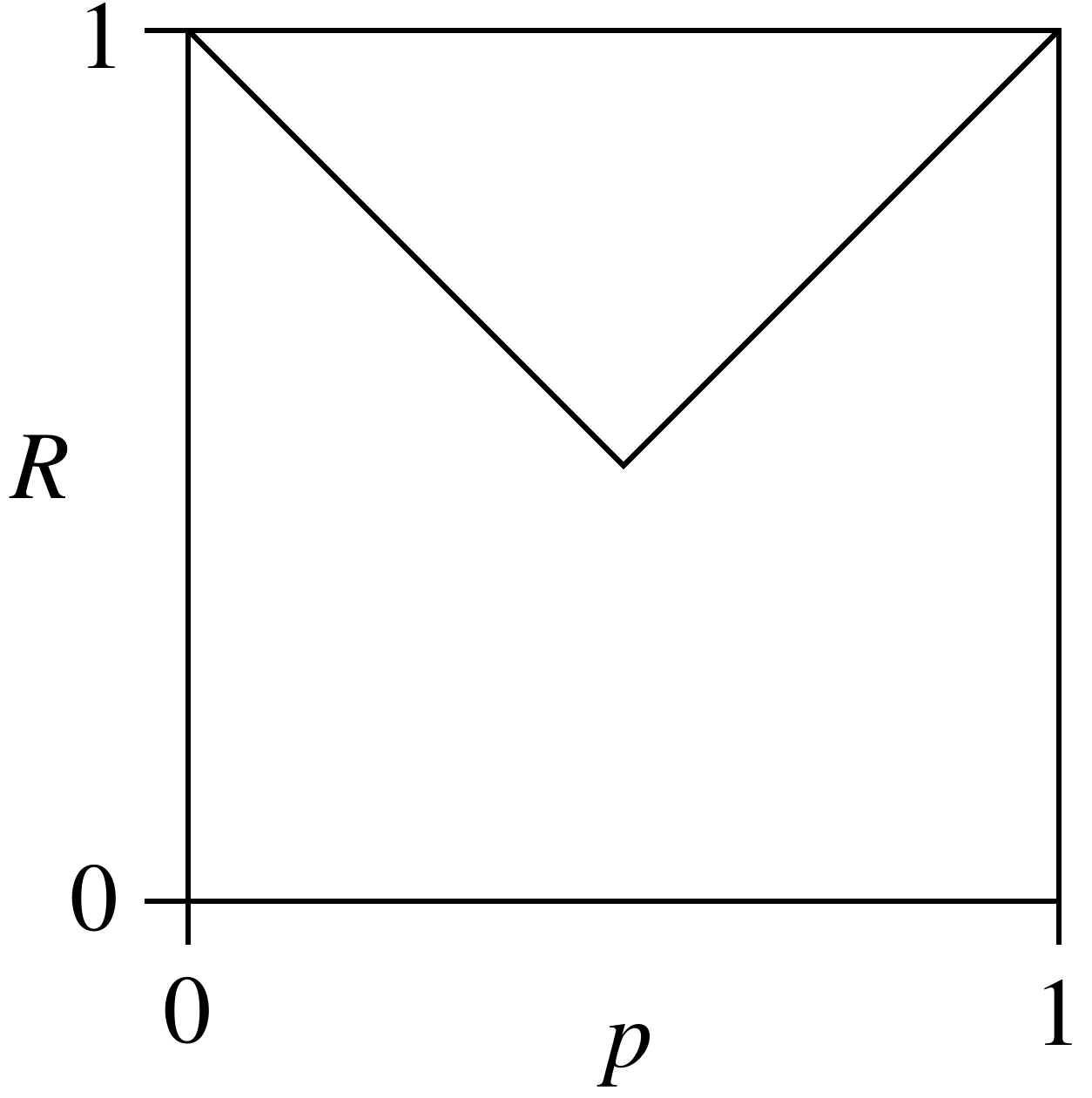}
\end{center}
\caption{
Graph of the reliability of blind guessing, $R_{\psi,p}(\E_\emptyset)$, as a function of $p$.}
\label{fig4}
\end{figure}

The adjective ``trivial'' is warranted in this case because this experiment requires no measurement or observation on the part of Alice---stretching, indeed, the notion of ``experiment.'' The result will be the same, independent of the actual state of the system. Because of this, the reliability is independent of any knowledge about the initial state. This yields a lower bound on $R^{\max}$,
\begin{equation}\label{RmaxbdE0}
R^{\max}_p(\psi)\geq \max(p,1-p)\,.
\end{equation}
Perhaps surprisingly, it is also sometimes an upper bound:

\begin{prop}\label{prop:no-better-than-guessing}
For $p \geq d/(d+1)$, no experiment is more reliable than blind guessing: $R^{\max}_p(\psi) = p$.
\end{prop}

This will follow from Thm.~\ref{thm:complete-information-optimality} below.


\subsection{Optimal Experiment}\label{sec:complete-information-experimentation}

Helstrom's theorem yields the optimal $E$ and the maximal reliability for given $\psi$ and $p$ as follows. Examples of $R^{\max}$ as a function of $p$ (for fixed $\psi$) are shown in Fig.~\ref{fig56}.

\begin{theorem}\label{thm:complete-information-optimality}
Let $0<p<1$ and $\psi\in\SSS$ with $\psi_k:=\scp{b_k}{\psi}\neq 0$ for all $k=1,\ldots,d$. Then
\begin{equation}\label{eqn:optimal-psi-reliability}
R^{\max}_p(\psi) = \begin{cases}
p & \text{if } p \geq d/(d+1) \\
p(1 + f^{-1}_\psi(\frac{p}{1-p})) & \text{if } p < d/(d+1) \end{cases},
\end{equation}
where $f_{\psi} : [0, \infty) \to (0,d]$ is the bijection given by
\begin{equation}\label{eqn:f-def}
f_\psi(z) = \sum_{k = 1}^d \frac{ |\psi_k|^2 }{z + |\psi_k|^2}.
\end{equation}
The optimal operators $E_\mathrm{opt}$ for which this maximum is attained, $R^{\max}_p(\psi) = R_{\psi, p}(E_\mathrm{opt})$, are
\begin{equation}\label{optimal-psi-E}
E_\mathrm{opt} = \begin{cases}
I & \text{if } p > d/(d+1), \\
I-\kappa\pr{\phi} & \text{if } p = d/(d+1), \\
I - \pr{\phi} & \text{if } p < d/(d+1), \end{cases}
\end{equation}
where $\kappa\in[0,1]$ is arbitrary and $\phi$ is the unique (up to a phase factor) normalized eigenvector of the unique non-positive eigenvalue of the operator
\begin{equation}\label{Adef}
A = p \bigl(\diag \pr{\psi}\bigr) - (1-p) \pr{\psi}.
\end{equation}
\end{theorem}

\begin{figure}[h]
\begin{center}
\includegraphics[width=5cm]{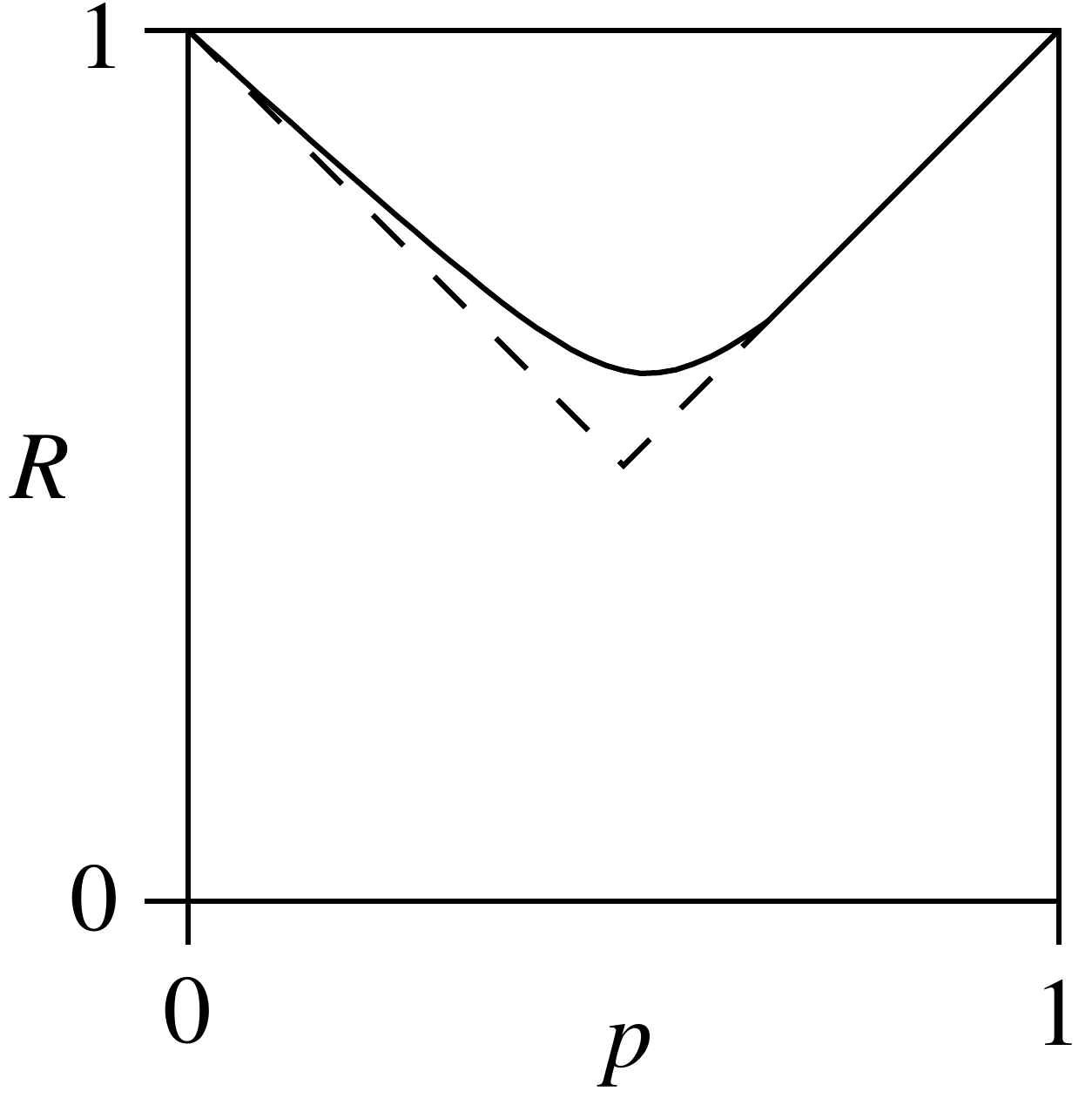}\quad\quad
\includegraphics[width=5cm]{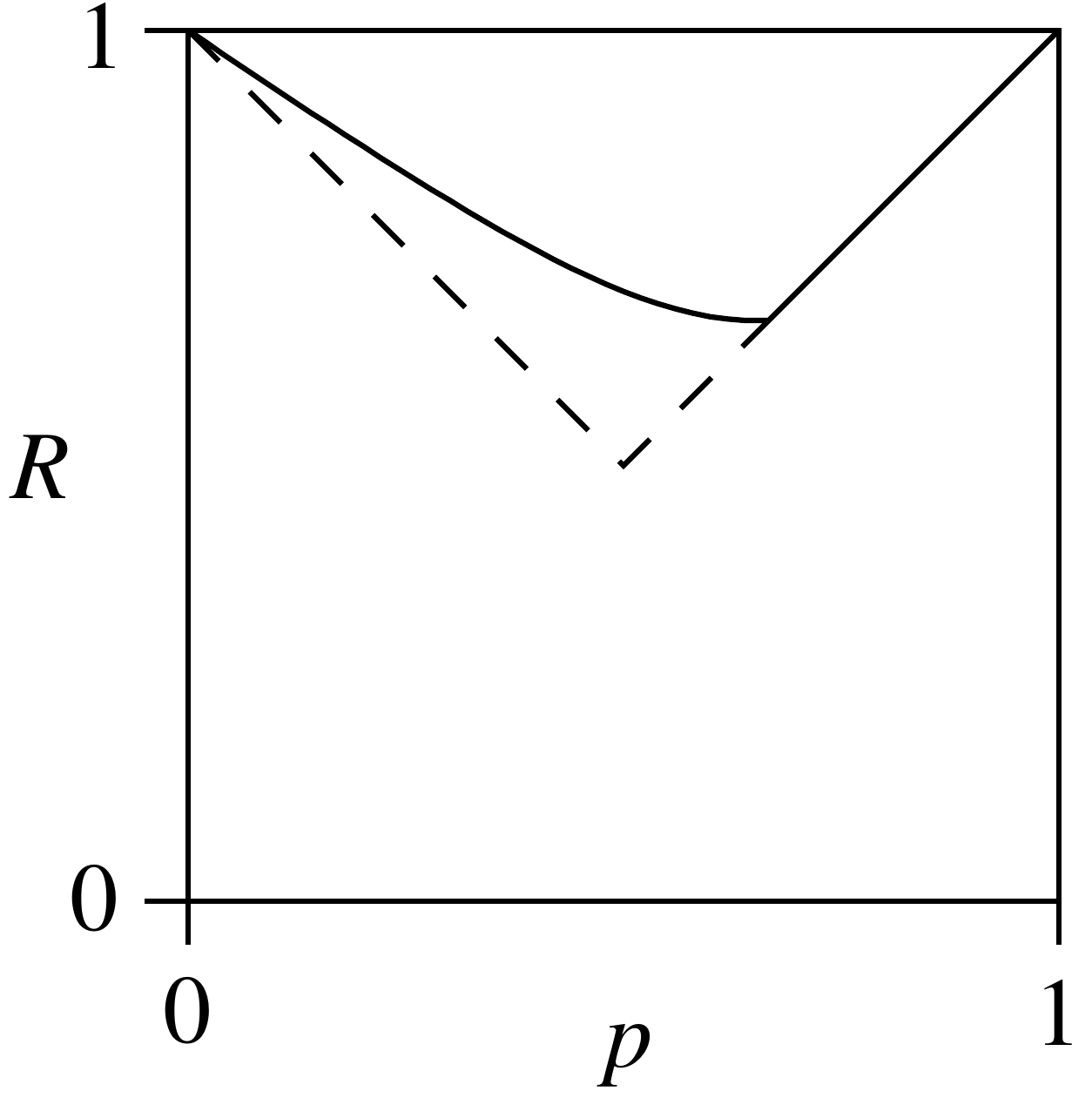}
\end{center}
\caption{
Examples of graphs of $R^{\max}_p(\psi)$ as a function of $p$ for (LEFT) $\psi= \sqrt{0.05}\,b_1+\sqrt{0.95}\,b_2$ and (RIGHT) $\psi=\sqrt{0.2}\,b_1+\sqrt{0.8}\,b_2$. Dashed lines: reliability of blind guessing, $R_{\psi,p}(\E_\emptyset)$, as a function of $p$.
}
\label{fig56}
\end{figure}

\medskip

\begin{proof}
In our situation, with $\rho_2=\pr{\psi}$ and $\rho_1=\diag\pr{\psi}$, the $A$ operator referred to in Helstrom's theorem and defined in \eqref{A-def-general} is just the one given by \eqref{Adef}. We first show that for $p > d/(d+1)$, $A$ has no non-positive eigenvalue, and for $p\leq d/(d+1)$ it has exactly one, which is non-degenerate and is $0$ for $p=d/(d+1)$ and negative for $p< d/(d+1)$.

Indeed, suppose that $\alpha$ is a non-positive eigenvalue of $A$, with eigenvector $\ket{\phi}$, $A \ket{\phi} = \alpha \ket{\phi}$. In that case, we have that 
\begin{equation}\label{eqn:psi-eigenvalue-structure-derivation}
\begin{split}
\Bigl(p \diag \pr{\psi} - (1-p) \pr{\psi}\Bigr) \ket{\phi} &= \alpha \ket{\phi} \\
-(1-p) \braket{\psi|\phi} \ket{\psi} &= -p\Bigl(-\frac{\alpha}{p} I + \diag \pr{\psi}\Bigr) \ket{\phi}
\end{split}
\end{equation}
Defining $M := -\frac{\alpha}{p} I + \diag \pr{\psi}$ 
and $\beta := (1-p) \braket{\psi|\phi}$, we have
\begin{equation}
- \beta \ket{\psi} = -p M \ket{\phi}.
\end{equation}
Note that $M$ is a diagonal matrix with strictly positive entries (as $-\alpha\geq 0$ and $|\psi_k|^2>0$), and is therefore invertible. As a result, $\beta \neq 0$, and $\phi$ and $\psi$ are not orthogonal. Moreover, we may write 
\begin{equation}\label{phibetapM-1psi}
\ket{\phi} = \frac{\beta}{p} M^{-1} \ket{\psi}.
\end{equation}
Hitting this with $\bra{\psi}$, and noting that $\braket{\psi|\phi} = \beta/(1-p)$, 
\begin{equation}
\frac{\beta}{1-p} = \frac{\beta}{p} \braket{\psi|M^{-1}|\psi},
\end{equation}
or 
\begin{equation}\label{alphacondition}
\begin{split}
\frac{p}{1-p} & = \braket{\psi|M^{-1}|\psi}\\
& = \sum_{k = 1}^d \frac{ |\psi_k|^2 }{ -\alpha/p + |\psi_k|^2 } \\
& = f_\psi( -\alpha/p ).
\end{split}
\end{equation}
For $z \geq 0$, $z\mapsto f_\psi(z)$ is continuous and stricly decreasing; its infimum is $\lim_{z\to \infty} f_\psi(z)=0$, its supremum is $f_\psi(0)=d$ (recall that we assumed that all $\psi_k\neq 0$); $f_\psi$ is thus a bijection $[0,\infty)\to (0,d]$. Hence, for $p/(1-p) \leq d$ (or, equivalently, $p\leq d/(d+1)$), the inverse $f_\psi^{-1}$ is well defined, and we have that the unique non-positive eigenvalue of $A$ is 
\begin{equation}\label{alphapfpsi}
\alpha = -p f^{-1}_\psi\Bigl( \frac{p}{1-p} \Bigr)\,,
\end{equation}
which is 0 for $p/(1-p)=d$ (i.e., $p=d/(d+1)$) and negative for $p/(1-p)<d$ (i.e., $p<d/(d+1)$). 
When $p/(1-p) > d$ (i.e., $p> d/(d+1)$), no solution to \eqref{alphacondition} and therefore no non-positive eigenvalue of $A$ exists.

Note further that for any fixed non-positive eigenvalue $\alpha$, any corresponding eigenvector $\phi$ must satisfy \eqref{phibetapM-1psi}. Hence, any eigenvector corresponding to $\alpha$ must lie on the same 1-dimensional subspace spanned by $M^{-1} \ket{\psi}$. As such, $\phi$ is unique up to scale, and as an eigenvalue of $A$, $\alpha$ has multiplicity $1$. 

Now, \eqref{eqn:optimal-psi-reliability} follows from \eqref{general-optimality} in Helstrom's theorem and \eqref{alphapfpsi}, and \eqref{optimal-psi-E} follows from \eqref{Eoptrho12} in Helstrom's theorem. This completes the proof.
\end{proof}

\noindent Remarks.

\begin{enumerate}
\item In case that $\psi_k=0$ for some $k$, the problem can be reduced to a subspace of smaller dimension, to which Thm.~\ref{thm:complete-information-optimality} can then be applied, except in the extreme case in which the dimension of the subspace is 1; in this case, which corresponds to $\psi$ being one of the $b_k$ (up to a phase), there is no difference between the collapsed and the uncollapsed state, and so, for a trivial reason, no experiment is more reliable than blind guessing.

\item As an example, consider the special $\psi$ with $\psi_k=1/\sqrt{d}$ for all $k=1,\ldots,d$, for $p<d/(d+1)$. It is readily verified that, in this case, $\phi=\psi$, as that is an eigenvector of $A$ with negative eigenvalue $p/d-(1-p)$. It follows that $E_\mathrm{opt}=I-\pr{\psi}$, which interestingly is \emph{independent of $p$}, and that $R^{\max}_p=1-p/d$.

\item We make the connection with perturbation theory of Hermitian matrices: It is a classical result that the spectrum of a Hermitian matrix is stable with respect to small perturbations (i.e., depends continuously on the matrix). For $p$ near 1, $A$ is dominated by $p \diag \pr{\psi}$, with a small perturbation consisting of $-(1-p)\pr{\psi}$. As such, for sufficiently large $p$, we expect the eigenvalues of $A$ to be effectively those of $\diag \pr{\psi}$---which are all positive. This agrees with the finding in the proof that all eigenvalues of $A$ are positive for large $p$, down to a transition point at $p = d/(d+1)$. Similarly, for small $p$, $A$ is dominated by $-(1-p)\pr{\psi}$, and since the perturbation $p\diag \pr{\psi}$ is a positive operator, perturbation theory tells us that $A$ has a single negative eigenvalue. 

\item Concerning the practical computation of $\phi$, and thus of $E_\mathrm{opt}$ for $p<d/(d+1)$, one may also, instead of finding the eigenspace of $A$ with negative eigenvalue, minimize $\braket{\phi|A|\phi}$. To simplify calculations, one may rotate the phases of the basis vectors $b_k$ so that $\psi_k\in[0,1]$ for $k = 1, \ldots, d$; note that such a change of $b_k$ has no effect on the distribution of the collapsed state vector $\psi'$. Then $A$ will have only real entries, and so will $\phi$ (up to a global phase that we can drop); so, we can take $\phi_k\in [0,1]$ for $k = 1, \ldots, d$ as well. This leads to the expression
\begin{equation}\label{eqn:opt-phi}
\braket{\phi|A|\phi}= p \left( \sum_{k = 1}^d \phi_k^2 \psi_k^2 \right) - (1-p) \left( \sum_{k = 1}^d \psi_k \phi_k \right)^{\!\!2}
\end{equation}
that needs to be minimized.
\end{enumerate}


\subsection{The Case of Dimension $d = 2$}\label{sec:complete-information-d=2}

The spin space of a spin-$\tfrac{1}{2}$ particle may serve as an example of a Hilbert space of dimension 2. We will identify $\Hilbert$ with $\CCC^2$ using the basis $\{b_1,b_2\}$. Spin space is equipped with a natural bijection between the 1D subspaces of $\CCC^2$ and the rays (or directions) in physical space $\mathbb{R}^3$, defined by the mapping $\CCC^2\to \RRR^3$, $\psi \mapsto \braket{\psi|\vsigma|\psi}$, with $\vsigma=(\sigma_1,\sigma_2,\sigma_3)$ the vector consisting of the 3 Pauli matrices
\be\label{Paulimatrices}
\sigma_1 = \begin{pmatrix}0&1\\1&0\end{pmatrix}\,\quad
\sigma_2= \begin{pmatrix}0&-i\\i&0\end{pmatrix}\,\quad
\sigma_3=\begin{pmatrix}1&0\\0&-1\end{pmatrix}\,,
\ee
for the appropriate choice of Cartesian coordinates in physical space. The basis vectors $b_1, b_2$ (or $(1,0), (0,1)$ in $\mathbb{C}^2$) then correspond to the positive and negative $z$-direction, respectively.

For $p\geq 2/3$, blind guessing is the optimal experiment. For $p< 2/3$ we can describe the optimal experiment as follows.

\begin{prop}\label{prop:dim2}
Let $d=2$, let $0<p< 2/3$, let $\psi\in\SSS$ with $\psi\notin \CCC b_1, \psi\notin \CCC b_2$, and let $\vv$ be the unit vector in the corresponding direction in $\RRR^3$, $\vv=(v_1,v_2,v_3)=\braket{\psi|\vsigma|\psi}$. Then $E_\mathrm{opt}=\pr{\chi}$ with
\be
\braket{\chi|\vsigma|\chi}\propto -\biggl( v_1,v_2,\Bigl(1-\frac{p}{1-p}\Bigr)v_3 \biggr)=:-\vw
\ee
with positive proportionality constant. That is, the optimal experiment is a Stern--Gerlach experiment in the direction $\vw$ obtained from $\vv$ by a dilation by the factor $1-\frac{p}{1-p}$ along the $z$ axis, with the outcome ``down'' labeled as ``yes'' and ``up'' labeled as ``no.''
\end{prop}

\begin{proof}
Change $b_1,b_2$ by phase factors so that $\psi_1,\psi_2$ are real and positive (using $\psi_1\neq 0 \neq \psi_2$), and rotate the Cartesian coordinate system in physical space so that \eqref{Paulimatrices} still holds (which is a rotation about the $z$ axis); then $\vv=(v_1,v_2,v_3)=(2\psi_1\psi_2,0,\psi_1^2-\psi_2^2)$. Let $q=1-p$ and $r=2p-1$, note that $-r/q=1-\frac{p}{1-p}$ and
\be
A=\begin{pmatrix} r\psi_1^2 & -q\psi_1\psi_2\\ -q\psi_1\psi_2 & r\psi_2^2\end{pmatrix},
\ee
and set $\vw=(w_1,w_2,w_3)=\bigl(v_1,0,-(r/q)v_3\bigr)$, $s_{\pm}:= \sqrt{\|\vw\|\pm w_3}$, $\tilde\chi = ( s_+,s_-)\in\CCC^2$, and $\tilde{\tilde{\chi}}=(-s_-,s_+)\in\CCC^2$; note that $\|\vw\|-w_3>0$ because $v_1>0$ due to $\psi_1,\psi_2>0$. A computation shows that
\be
\Delta:=\frac{(A\tilde\chi)_1}{\tilde\chi_1}-\frac{(A\tilde\chi)_2}{\tilde\chi_2} =r(\psi_1^2-\psi_2^2)+q\psi_1\psi_2\Bigl(\frac{s_+}{s_-}-\frac{s_-}{s_+}\Bigr).
\ee
One verifies that $s_+/s_--s_-/s_+= 2w_3/w_1$, so $\Delta=0$, which shows that $\tilde{\chi}$ is an eigenvector of $A$ with eigenvalue $\tilde\alpha=(A\tilde\chi_1)/\tilde\chi_1$. Since $\tilde{\tilde{\chi}}$ is orthogonal to $\tilde\chi$, or by a similar computation, $\tilde{\tilde{\chi}}$ is also an eigenvector with eigenvalue $\tilde{\tilde{\alpha}}$. One verifies that $\tilde{\alpha}<\tilde{\tilde{\alpha}}$, and since we know that $A$ has exactly one negative eigenvalue (for $p<2/3$), we must have $\tilde{\alpha}<0\leq \tilde{\tilde{\alpha}}$. Thus, $\phi$ must be proportional to $\tilde{\chi}$, so $E_\mathrm{opt}=I-\pr{\phi} = \pr{\chi}$ with $\chi=\tilde{\tilde{\chi}}/\|\tilde{\tilde{\chi}}\|$. One verifies that $\braket{\tilde{\chi}|\vsigma|\tilde{\chi}} = 2\vw$ and $\braket{\tilde{\tilde{\chi}}|\vsigma|\tilde{\tilde{\chi}}}=-2\vw$. Now rotate back the Cartesian coordinates and the basis $\{b_1,b_2\}$.
\end{proof}

We note for the sake of completeness that for $d=2$ the maximal reliability is
\begin{equation}\label{Rmax2D}
R^{\max}_p(\psi) = \begin{cases}
p & \text{if } p \geq 2/3 \\
\frac{1}{2} + \frac{1}{2} \sqrt{ (1-2p)^2 + 4p(2-3p) |\psi_1|^2 |\psi_2|^2  } & \text{if } p < 2/3 \end{cases}.
\end{equation}
Graphs of $p\mapsto R^{\max}_p(\psi)$ are shown in Fig.~\ref{fig56} for two different choices of $\psi$; the graph of $|\psi_1|^2\mapsto R^{\max}_p(\psi)$ (with $|\psi_2|^2=1-|\psi_1|^2$) is shown in Fig.~\ref{fig7} for $p=1/2$.

\begin{figure}[h]
\begin{center}
\includegraphics[width=5cm]{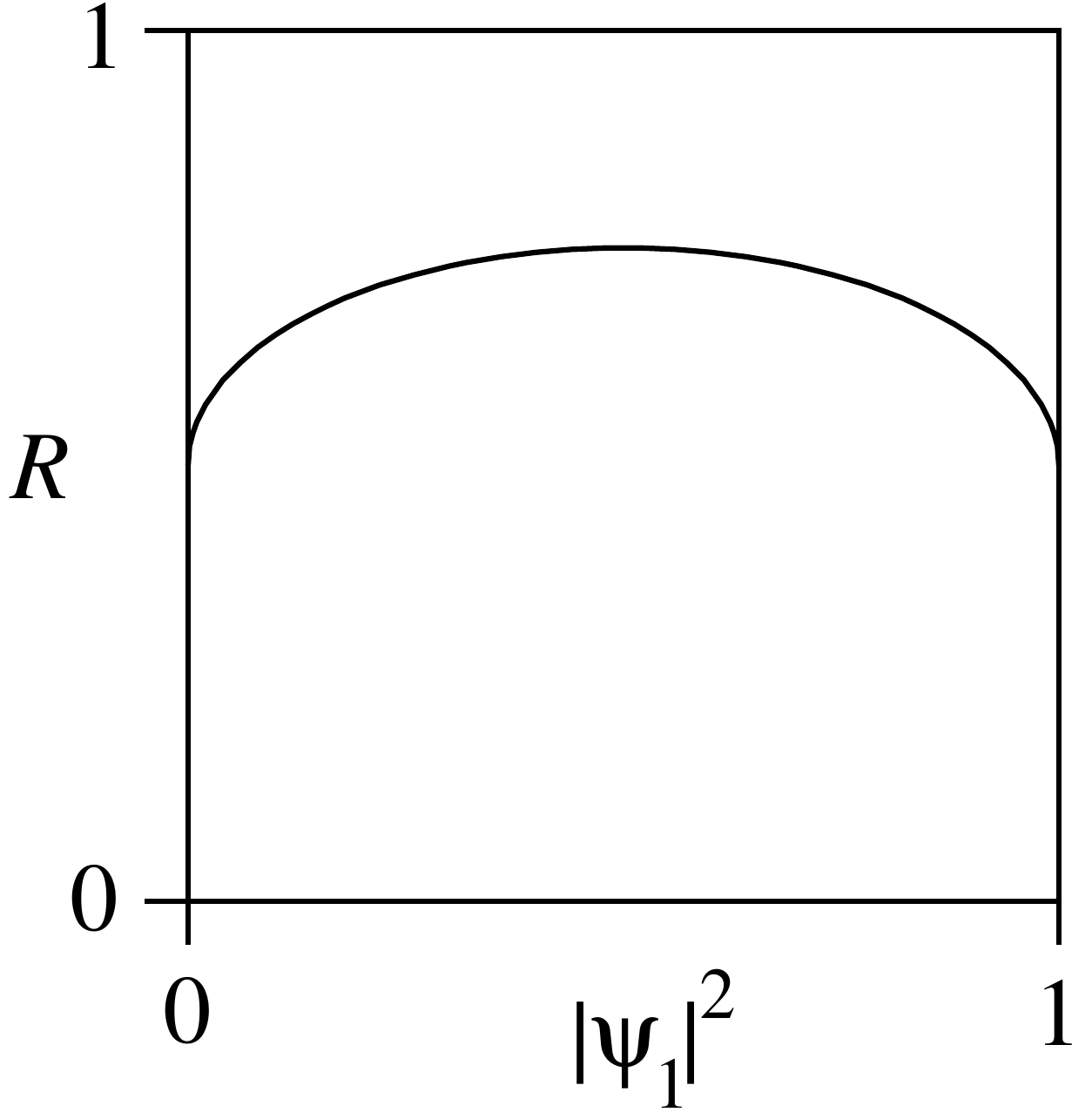}
\end{center}
\caption{
Graph of the maximal reliability $R^{\max}_{p}(\psi)$ as in \eqref{Rmax2D} for $d=2$ (i.e., $\psi\in\CCC^2$) and $p=1/2$, as a function of $|\psi_1|^2$. The shape is the upper half of an ellipse centered at $(\frac{1}{2},\frac{1}{2})$, the maximal value is 3/4.
}
\label{fig7}
\end{figure}


\subsection{Bounds on the Maximal Reliability}
\label{sec:bounds}

While Thm.~\ref{thm:complete-information-optimality} specifies the value of $R^{\max}_p(\psi)$, it is sometimes useful to have bounds on $R^{\max}_p(\psi)$ that are easier to compute. Some of the following bounds are depicted in Fig.~\ref{fig5c6c}.

\begin{corollary}\label{cor:r-max-bounds}
For all $\psi\in\SSS$ and $0<p<1$,
\begin{equation}\label{eqn:r-max-bounds}
\max\Bigl(p, 1 - p \sum_{k = 1}^d |\psi_k|^4\Bigr) \leq R^{\max}_p(\psi) \leq \max(p, 1-p/d).
\end{equation}
\end{corollary}

\begin{figure}[h]
\begin{center}
\includegraphics[width=5cm]{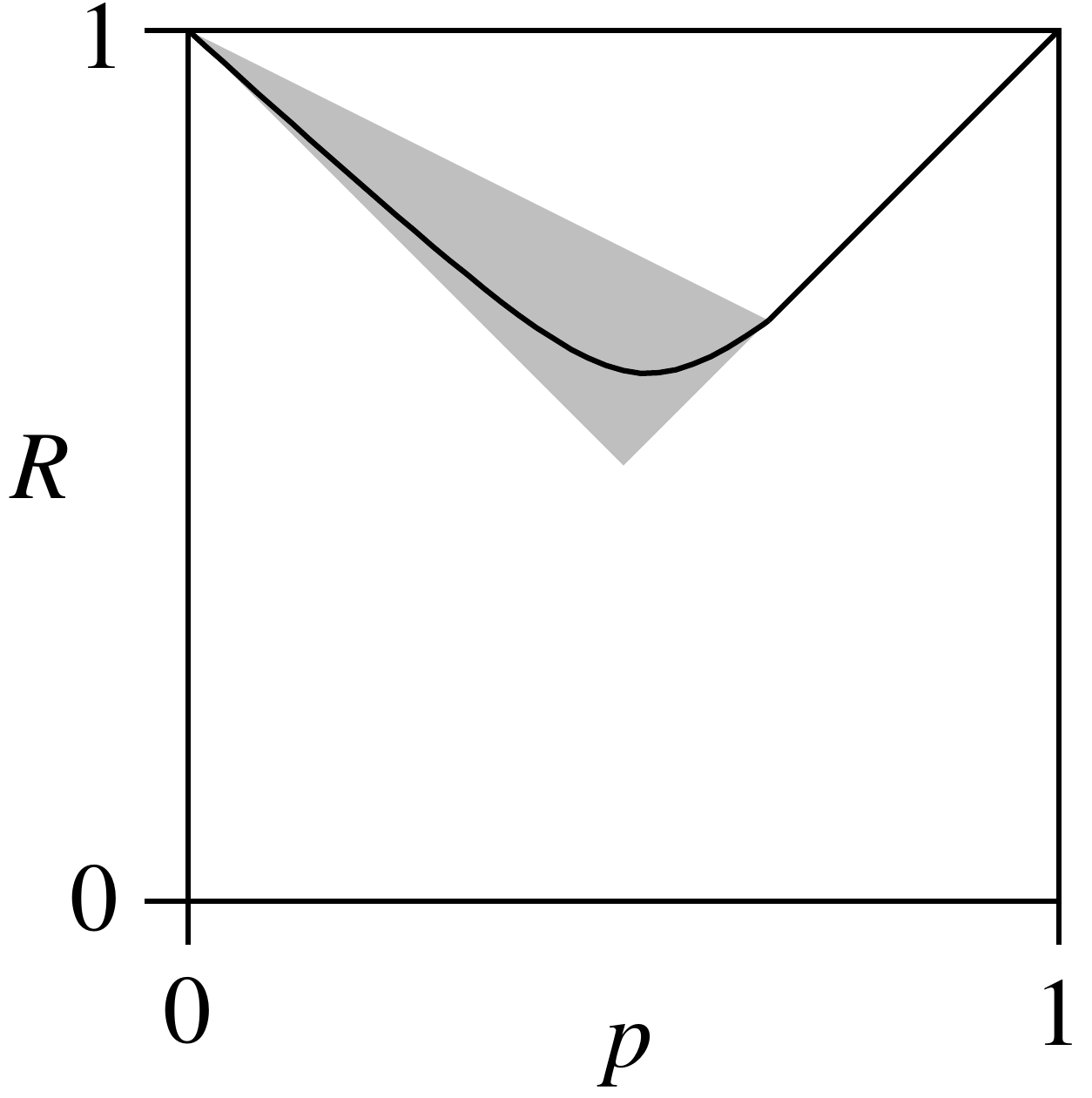}\quad\quad
\includegraphics[width=5cm]{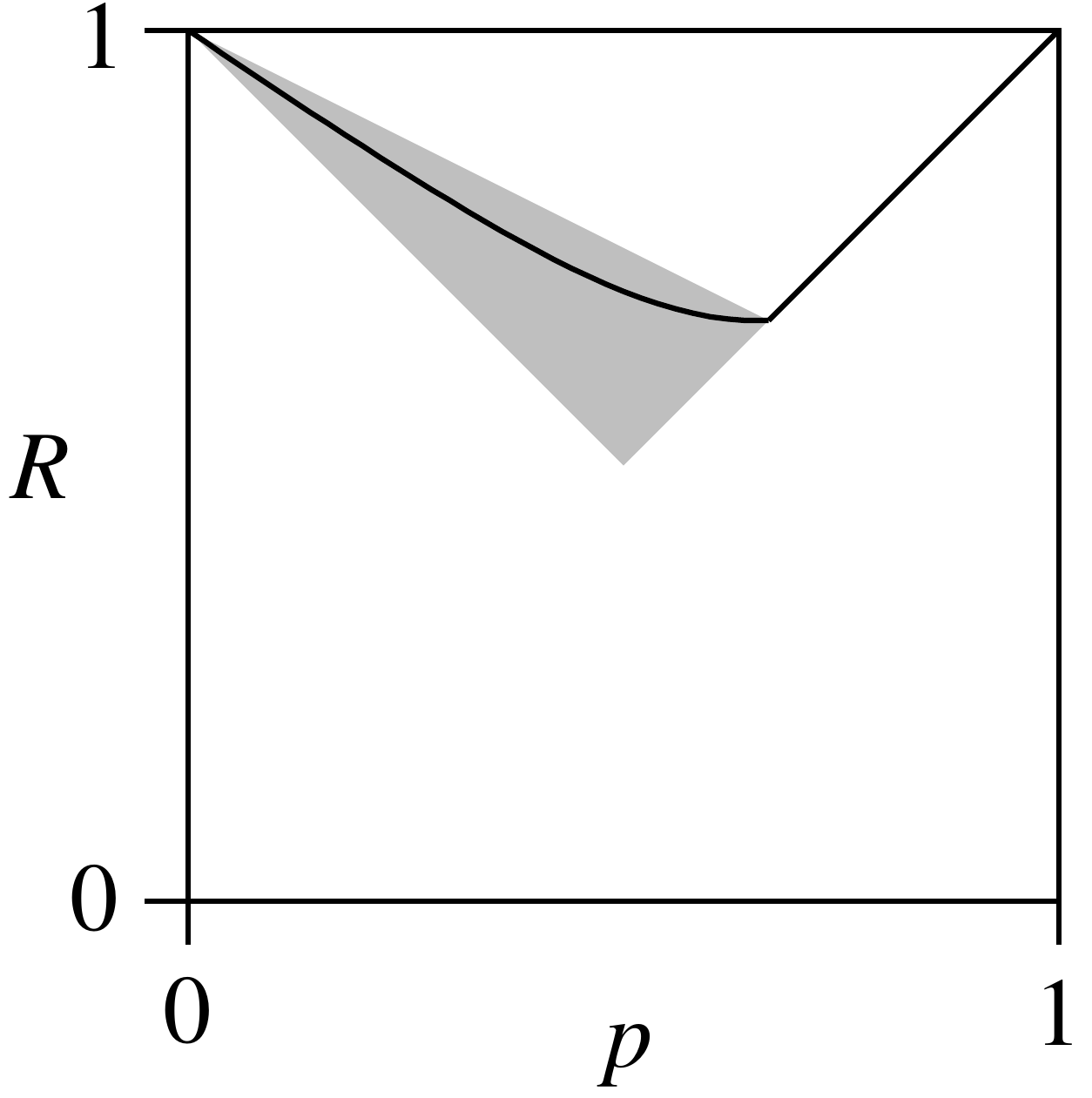}
\end{center}
\caption{
The same two graphs as in Fig.~\ref{fig56}, along with the shaded region characterized by the $\psi$-independent bounds provided by Cor.~\ref{cor:r-max-bounds}, $\max(p,1-p)\leq R^{\max}_p(\psi)\leq \max(p,1-p/d)$. For $\psi$-dependent bounds, see Fig.~\ref{fig5d6d}.}
\label{fig5c6c}
\end{figure}

\begin{proof}
We begin by noting the property of the function $f_\psi(z)$ defined by \eqref{eqn:f-def} that for any fixed $z>0$, $f_\psi(z)$ increases if we change two of the $|\psi_k|^2$ so that the smaller grows and the bigger shrinks (while the sum remains constant). Indeed, suppose $|\psi_1|^2<|\psi_2|^2$ and change $|\psi_1|^2\to |\psi_1|^2+dx$, $|\psi_2|^2\to |\psi_2|^2-dx$ with infinitesimal $dx>0$; then, to first order in $dx$ and leaving aside the unchanged terms with $k>2$,
\begin{equation}
\begin{split}
&\frac{|\psi_1|^2+dx}{z+|\psi_1|^2+dx} + \frac{|\psi_2|^2-dx}{z+|\psi_2|^2-dx}\\
&=\frac{|\psi_1|^2}{z+|\psi_1|^2} + \frac{z}{(z+|\psi_1|^2)^2}dx
+\frac{|\psi_1|^2}{z+|\psi_2|^2} -  \frac{z}{(z+|\psi_2|^2)^2}dx\\
&>\frac{|\psi_1|^2}{z+|\psi_1|^2}+\frac{|\psi_1|^2}{z+|\psi_2|^2} \,.
\end{split}
\end{equation}
As a consequence,
\be\label{maxf}
\max_{\psi\in\SSS} f_\psi(z) = \frac{1}{z+1/d} \text{ and }
\min_{\psi\in\SSS} f_\psi(z) = \frac{1}{z+1}
\ee
with the maximum attained at $\psi_k =1/\sqrt{d}$ and the minimum at $\psi_k=\delta_{k1}$.

To verify the upper bound in \eqref{eqn:r-max-bounds}, we note that if $f\leq g$ for decreasing functions then $f^{-1}\leq g^{-1}$, so with $g_\psi(z)=1/(z+1/d)$ the first equation of \eqref{maxf} yields
\be
f_\psi^{-1}(u) \leq \frac{1}{u}-\frac{1}{d}\,,
\ee
which with \eqref{eqn:optimal-psi-reliability} gives the upper bound in \eqref{eqn:r-max-bounds}.

The lower bound can be derived from the fact that the harmonic mean is always less than or equal to the arithmetic mean, which implies that
\be
f_\psi(z) \geq \frac{1}{z+ \sum |\psi_k|^4}\,.
\ee

A more illustrative proof for the lower bound goes as follows. 
Choosing $E = I - \pr{\psi}$ yields
\begin{equation}
\begin{split}
R_{\psi, p}(E) & = (1-p) + \tr \left[E A\right] \\
& = (1-p) + p \tr \left[(I - \pr{\psi}) \diag \pr{\psi} \right] \\
& = (1-p) + p\Bigl(1 - \sum_{k = 1}^d |\psi_k|^4 \Bigr)\\
& = 1 - p \sum_{k = 1}^d |\psi_k|^4.
\end{split}
\end{equation}
This choice, or else blind guessing whenever that is more reliable, gives the desired lower bound. 
\end{proof}

\noindent Remarks.

\begin{enumerate}
\item Since under the assumption $\psi_k\neq 0$, $\sum |\psi_k|^4 < 1$, the lower bound in \eqref{eqn:r-max-bounds} is an improvement on the lower bound \eqref{RmaxbdE0} provided by blind guessing alone.

\item The upper and lower bounds of Cor.~\ref{cor:r-max-bounds} are tight in the sense that equality holds for some $\psi$. Indeed, setting $\psi_k=1/\sqrt{d}$ for all $k=1,\ldots,d$, the lower bound coincides with the upper bound, and $R^{\max}_p(\psi) = 1 - p/d$.

\item It follows further that, for any fixed $0<p<1$, $R^{\max}_p(\psi)$ can be made arbitrarily close to 1 for suitable choice of $d$ and $\psi$. However, this does not mean that for large $d$ it be typical for $\psi$ to have $R^{\max}_p(\psi)$ close to 1. The situation is analyzed further in the following two remarks.

\item The following bound is similar to \eqref{maxf} but slightly tighter:
For any $\psi\in\SSS$, let $\delta = \max_k |\psi_k|^2$. Then,
\begin{equation}\label{another-f-bound}
f_\psi(z) \leq \frac{\delta}{z + \delta} + \frac{1-\delta}{z + (1-\delta)/(d-1) }.
\end{equation}
Indeed, fixing the component with $|\psi_k|^2 = \delta$, $f_\psi(z)$ is maximized by equally distributing the remaining weight among the other components, $|\psi_j|^2 = (1-\delta)/(d-1)$ for $j \neq k$. This yields \eqref{another-f-bound}. This bound can be used to give an upper bound on $R^{\max}_p(\psi)$ that is tighter than the upper bound of Cor.~\ref{cor:r-max-bounds} (as the latter does not depend on $\psi$ but the former does through $\delta$). Furthermore, taking $d$ to infinity gives the following dimension-independent bound:
\begin{equation}\label{infinity-bound}
R^{\max}_p(\psi) \leq \frac{1}{2}\left( 1 + p(1-\delta) + \sqrt{ {(1-p)}^2 + 2 p (1-p)\delta - (4-5p)p \delta^2 }\right).
\end{equation}
This bound, depicted in Fig.~\ref{fig5d6d}, is of interest insofar as it is strictly less than 1 while valid for all $\psi$ with $\delta > 0$, even as $d\to\infty$. For instance, if $p = 1/2$ (corresponding in some sense to maximal initial uncertainty as to whether or not collapse has occurred), $R^{\max}_{1/2}(\psi) \leq 0.91$ whenever $\delta \geq 1/2$, independently of $d$.

\begin{figure}[h]
\begin{center}
\includegraphics[width=5cm]{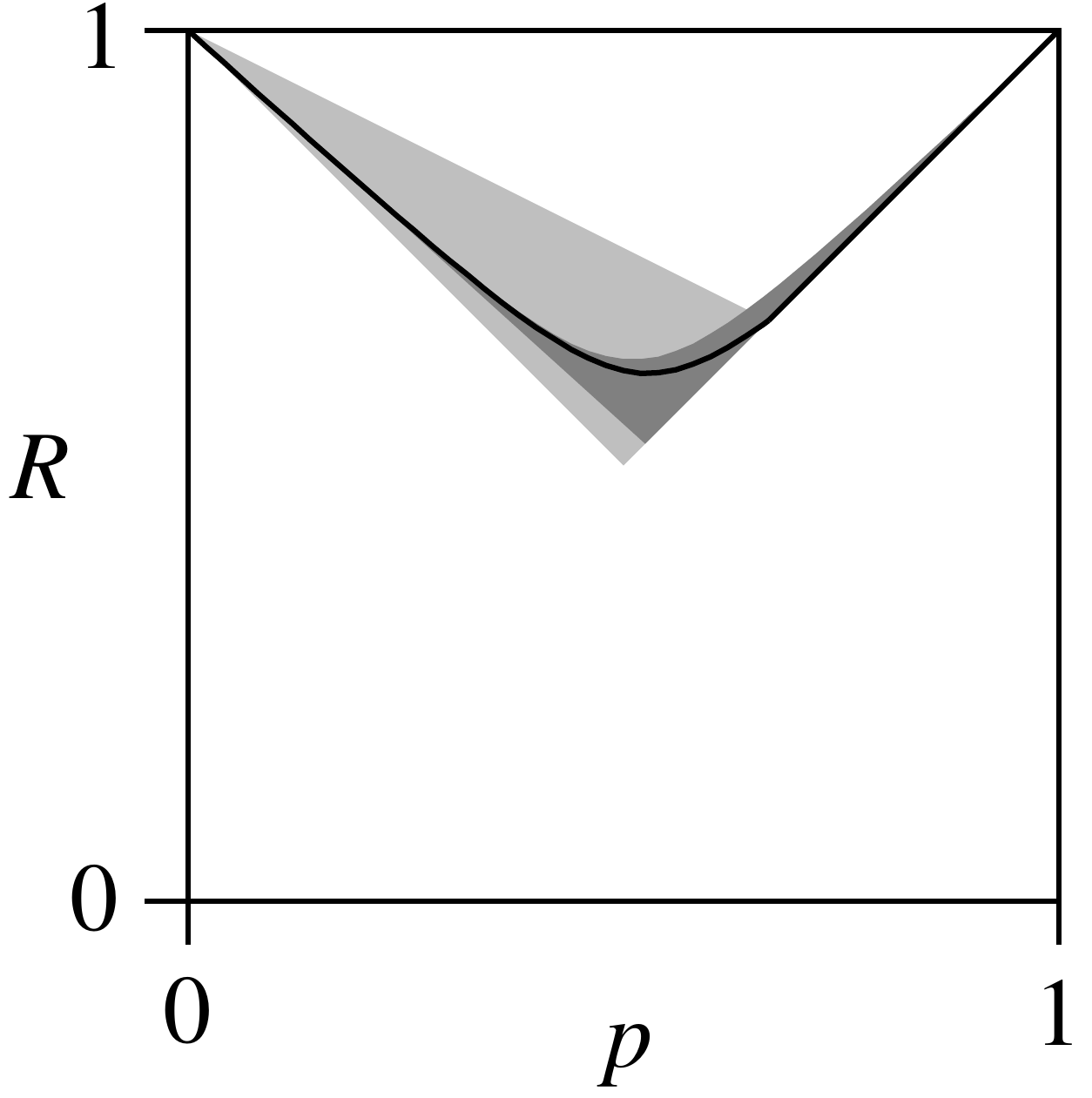}\quad\quad
\includegraphics[width=5cm]{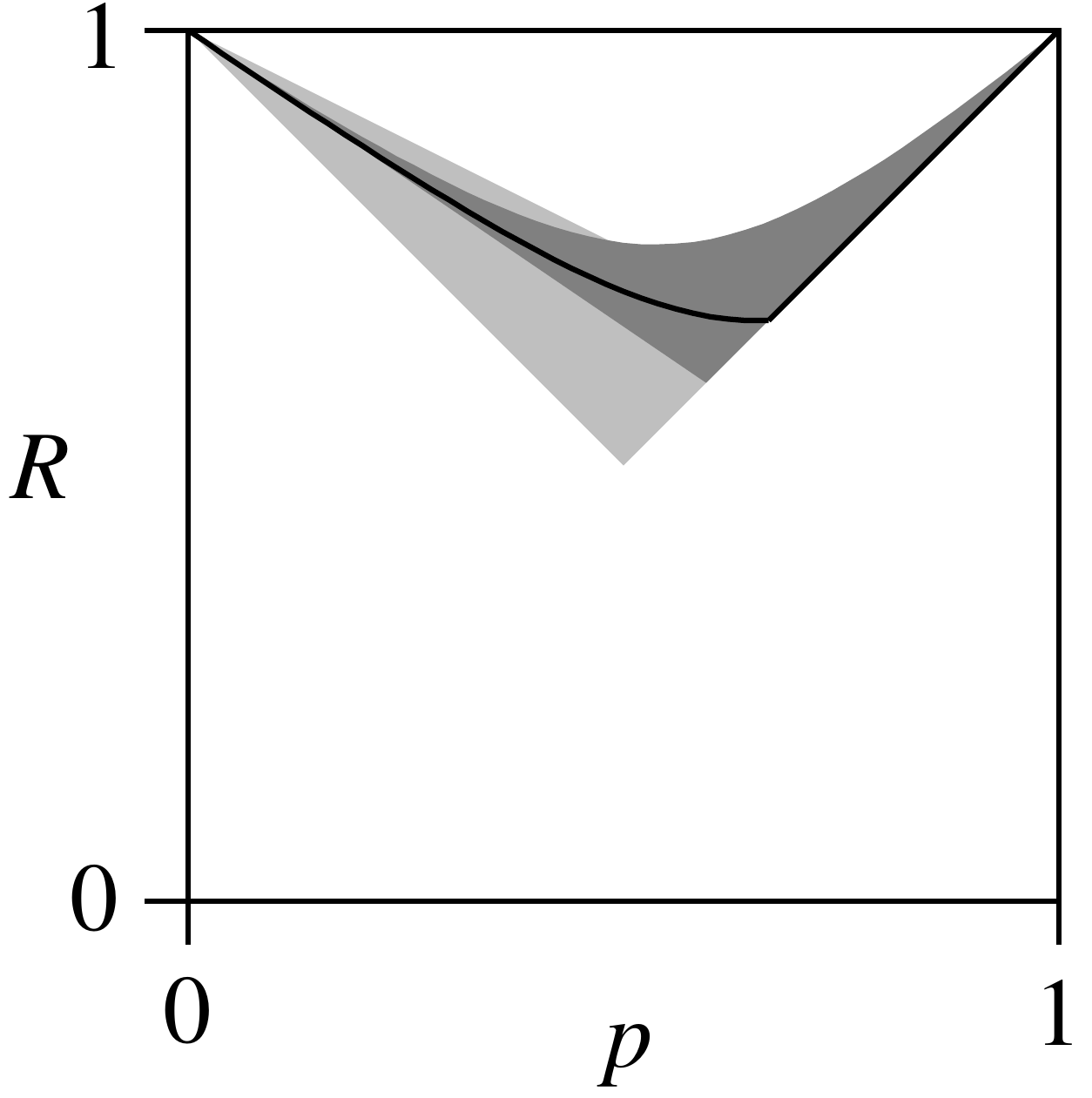}
\end{center}
\caption{
The same two diagrams as in Fig.~\ref{fig5c6c}, and in addition the darkly shaded region characterized by our $\psi$-dependent bounds: the upper bound provided by \eqref{infinity-bound} and the lower bound provided by Cor.~\ref{cor:r-max-bounds}, $\max(p,1-p\sum |\psi_k|^4)\leq R^{\max}_p(\psi)$.
}
\label{fig5d6d}
\end{figure}

\item Another remark concerns the $\delta$-dependence of the bound \eqref{infinity-bound}.
If, for a sequence of $\psi$s with $d\to\infty$, $\delta = \max |\psi_k|^2$ tends to $0$, as in the case with $|\psi_k|^2 = 1/d$ for each $k$, then $R^{\max}_p(\psi)$ tends to $1$ in the limit. Note that, in this situation, also the lower bound of Cor.~\ref{cor:r-max-bounds} approaches $1$.

Alternately, if $\delta$ tends to $1$, so that the weight gets concentrated on a single component, even as $d$ increases, then $R^{\max}_p(\psi)$ tends to $\max(p, 1-p)$, which coincides with the reliability of blind guessing; of course, this behavior is consistent with the lower bound given in Cor.~\ref{cor:r-max-bounds}. Concentrating the weight on a single component, while shrinking the other components to zero, effectively reduces the dimension relevant to the problem.

Alice is therefore in the best position if $\psi$ is such that the weights are distributed relatively uniformly across many components. As the largest weight increases and approaches $1$, the probability of correctly determining whether collapse has occurred diminishes; this is only expected because if $\psi$ is close to one of the $b_k$, its primary mode of collapse ($b_k$ times a phase) will be largely indistinguishable from its initial state.
\end{enumerate}


\section{Incomplete Information}
\label{sec:incomplete}

We now assume that Alice does not know $\psi$ precisely, but only knows that the initial state of $S$ was drawn from a known distribution $\mu$. The situation of the previous section, with known $\psi$, is included in that $\mu$ may be a delta distribution on $\SSS$. We write $\psi\sim \mu$ to express that the random variable $\psi$ has distribution $\mu$, and $\EEE$ for expectation. As such, the results of this section parallel those of the previous. 

\subsection{Reliability and Optimal Experiment}
\label{sec:incomplete-information-reliability}

The reliability of a yes-no experiment, still defined to be the probability of correctly answering whether a collapse has occurred, is now a function of $\mu$ (instead of $\psi$), found to be
\begin{equation}
\label{eqn:limited-information-reliability}
\begin{split}
R_{\mu, p}(\E) & = \PPP(Z = \text{yes}, \text{collapse}) + \PPP(Z = \text{no}, \text{no collapse}) \\
& = \EEE_{\psi \sim \mu} \left[  \PPP(\{ Z = \text{yes}, \text{collapse} \} \text{ or } \{Z = \text{no}, \text{no collapse}\} | \psi)  \right] \\
& = \EEE_{\psi \sim \mu} \left[ R_{\psi, p}(\E) \right] \\
& = \int_{\SSS}\mu(d\psi) \,  \tr \left[ \pr{\psi} (p \diag E + (1-p) (I - E ))  \right] \\
& = \tr \left[ \rho\left( p\diag E + (1-p)(I-E)\right) \right]\,,
\end{split}
\end{equation}
where $\rho$ is the density matrix associated with distribution $\mu$, defined by
\begin{equation}\label{eqn:rho-definition-eqn}
\rho = \int_{\SSS} \mu(d\psi) \, \pr{\psi}\,.
\end{equation}
Note that in this case, the reliability depends on the distribution $\mu$ only through $\rho$: For two distributions $\mu_1\neq \mu_2$ with the same $\rho$, any experiment will produce equally reliable results on either distribution. We can thus write $R_{\rho,p}(E)$ instead of $R_{\mu,p}(\E)$. This observation also shows that it is not necessary for Alice to know $\mu$, it suffices to know $\rho$; and since $\rho$, but not $\mu$, can be measured if a large ensemble of systems is provided whose wave functions have distribution $\mu$, the assumption that Alice knows $\rho$ is natural. 

The statement analogous to Thm.~\ref{thm:complete-information-imperfect-reliability} is also true:

\begin{theorem}\label{thm:incomplete-information-imperfect-reliability}
For $0<p<1$ and any density matrix $\rho$, $R_{\rho,p}(E)<1$ for all $0\leq E\leq I$. That is, 
there is no experiment $\E$ that can correctly determine with probability $1$ whether or not $S$ has collapsed.
\end{theorem}

\begin{proof}
Choose any distribution $\mu$ with $\rho_\mu=\rho$; this is possible for every density matrix. Then the statement is immediate from the third line of \eqref{eqn:limited-information-reliability} and Thm.~\ref{thm:complete-information-imperfect-reliability}. (That is, having less information about the initial state $\psi$ cannot be conducive to having greater reliability.)
\end{proof}

In parallel with \eqref{eqn:r-max-psi}, we define
\begin{equation}\label{eqn:r-max-rho}
R^{\max}_p(\rho) = \sup_{0 \leq E \leq I} R_{\rho, p}(E)\,.
\end{equation}
As an immediate consequence of Cor.~\ref{cor:r-max-bounds}, we obtain the following:

\begin{corollary}\label{cor:r-max-bounds-rho}
For all density matrices $\rho$, and $0 < p < 1$,
\begin{equation}\label{eqn:r-max-bounds-rho}
R^{\max}_p(\rho) \leq \max(p, 1-p/d).
\end{equation}
\end{corollary}


\label{sec:incomplete-information-experimentation}

The maximal reliability and optimal $E$ are provided by Helstrom's theorem with
\begin{equation}\label{A-rho-def}
A = p (\diag \rho) - (1-p) \rho.
\end{equation}

\begin{prop}\label{prop:rhovsblind}
For $p\geq d/(d+1)$, $A$ never has a negative eigenvalue, so $E=I$ is optimal and $R^{\max}_p(\rho)=p$; that is, no experiment is more reliable than blind guessing.
\end{prop}

\begin{proof}
Let $K=\bigl\{k\in\{1,\ldots,d\}: b_k \notin \ker \rho\bigr\}$. We prove that blind guessing is optimal for $\# K/(\#K+1)\leq p <1$, a range including $d/(d+1)\leq p<1$. Since the dimensions $b_k$ with $k\notin K$ play no role in the problem, we can focus on the space $\mathrm{span}\{b_k:k\in K\}$, call that $\Hilbert$ in the remainder of this proof, and take $\#K=d$. The kernel of $\rho$ now does not contain any $b_k$, but it can still be nontrivial. Choose a probability distribution $\mu$ on $\SSS$ with $\rho_\mu=\rho$ that is absolutely continuous (relative to the uniform distribution) on the unit sphere in the positive spectral subspace of $\rho$ (for example, one such distribution is the ``Scrooge measure'' \cite{JRW94} with density matrix $\rho$). Then every coordinate hyperplane $H_k=\{\psi\in\Hilbert: \psi_k=0\}$ is a null set, $\mu(H_k)=0$, and Thm.~\ref{thm:complete-information-optimality} applies to a $\mu$-distributed $\psi$ with probability 1. Thus, for any $0\leq E\leq I$,
\begin{equation}\label{Rleqp}
R_{\rho,p}(E)=R_{\mu, p}(E) = \EEE_{\psi \sim \mu} \left[ R_{\psi, p}(E) \right] 
 \leq \EEE_{\psi \sim \mu} \left[ R^{\max}_p(\psi) \right ] 
= \EEE_{\psi \sim \mu} \left[ p \right] 
= p\,.
\end{equation}
(That is, again, having less information about the initial state $\psi$ cannot be conducive to having greater reliability.) 
Since the bound \eqref{Rleqp} is attained by $E = I$, that is an optimal $E$, and $R^{\max}_p(\rho)=p$. It also follows that, for $p \geq d/(d+1)$, $A$ has no negative eigenvalues.
\end{proof}

Again, we note the connection to perturbation theory of Hermitian matrices: For $p$ sufficiently close to 1, $A$ is dominated by $p \diag \rho$; if $\rho$ is of full rank, or at least none of the $b_k$ lies in the kernel of $\rho$, then $\diag \rho$ is of full rank, and perturbation theory implies that $A$ has only positive eigenvalues, so that the positive spectral subspace of $A$ is all of $\Hilbert$, and blind guessing is the unique optimal experiment. The proof above shows that, in fact, blind guessing is optimal for all $p\geq  d/(d+1)$.

For middle values of $p$, since $A$ is a combination of the positive operator $p \diag \rho$ and the negative operator $-(1-p)\rho$, we may expect that $A$ has both positive and negative eigenvalues, leading to a non-trivial behavior of $E_\mathrm{opt}$.

For small $p$, we expect $A$ to be dominated by $-(1-p) \rho$. In the case of Sec.~\ref{sec:complete-information-experimentation} with $\rho = \pr{\psi}$, $\rho$ had a single positive eigenvalue and a $(d-1)$-fold eigenvalue 0. Hence, and since $p\diag \rho$ is positive, it followed that, for small $p$, $A$ had a single negative eigenvalue. Now, however, we consider $\rho$ more broadly. In the (generic) case that $\rho$ is of full rank, $\rho$ will have \emph{only} positive eigenvalues, and then perturbation theory implies that, for sufficiently small $p$, $A$ has \emph{only} negative eigenvalues. A more specific statement is provided by the following proposition.

\begin{prop}\label{rho-trivial-lower-bound}
If $\rho$ has full rank with smallest eigenvalue $p_d>0$, and if $p\leq p'$ with
\begin{equation}\label{p-primedef}
p'= \frac{p_d}{\max_k \braket{b_k|\rho|b_k} + p_d}\,,
\end{equation}
then $E_\mathrm{opt} = 0$, i.e., blind guessing is an optimal experiment. We note that $p'\leq 1/2$.
\end{prop}

\begin{proof}
Given that $E_\mathrm{opt} = P^+_A$ for all $p$, it suffices to show that for $p\leq p'$ the eigenvalues of $A$ are all non-positive. Let $\lambda_1(M)$ denote the largest eigenvalue of the Hermitian matrix $M$. It is a classical result that for any Hermitian matrices $L, M$,
\begin{equation}\label{eigen-bound}
\lambda_1(L + M) \leq \lambda_1(L) + \lambda_1(M).
\end{equation}
(Indeed, this follows from $\lambda_1(M)=\max_{\psi\in\SSS}\braket{\psi|M|\psi}$.)
Applying \eqref{eigen-bound} to $A$ as in \eqref{A-rho-def},
\be
\begin{split}
\lambda_1(A) 
&\leq p \lambda_1(\diag \rho) + (1-p) \lambda_1( -\rho)\\
&= p\max_k \braket{b_k|\rho|b_k} - (1-p) p_d\leq 0
\end{split}
\ee
whenever $p\leq p'$ as in \eqref{p-primedef}. 
The last statement follows from
\be
\max_k \braket{b_k|\rho|b_k} \geq \min_{\psi\in\SSS} \braket{\psi|\rho|\psi} = p_d\,.
\ee
\end{proof}

A similar reasoning applies in the general setting of Helstrom's theorem, where
\begin{equation}
A = p \rho_1 - (1-p) \rho_2\,.
\end{equation}
For $p$ close to 1 we expect $A$ to be dominated by $p \rho_1$, and for small $p$, to be dominated by $-(1-p)\rho_2$. If $\rho_1$ and $\rho_2$ have full rank, then $A$ will have only positive eigenvalues for sufficiently large $p$, and only negative eigenvalues for sufficiently small $p$. Hence, in the case of $\rho_1, \rho_2$ with full rank, for all sufficiently large or sufficiently small $p$, blind guessing is the optimal experiment.
A more specific statement is provided by the following generalization of Prop.~\ref{rho-trivial-lower-bound}.

\begin{prop}
For any Hermitian $d\times d$ matrix $M$, let $\lambda_1(M)$ and $\lambda_d(M)$ denote the largest and smallest eigenvalues of $M$, respectively. For $\rho_1, \rho_2$ of full rank,
$E_\mathrm{opt} = 0$ for all $p \leq p'$ and $E_\mathrm{opt} = I$ for all $p \geq p''$, where
\begin{align}
p'&= \frac{ \lambda_d(\rho_2) }{ \lambda_1(\rho_1) + \lambda_d(\rho_2) },\\
p''&= \frac{ \lambda_1(\rho_2) }{ \lambda_d(\rho_1) + \lambda_1(\rho_2) }.
\end{align}
We note that $p'\leq \frac{1}{2}\leq p''$.
\end{prop}

\begin{proof}
As before, it suffices to show that the eigenvalues of $A$ are non-positive for all $p\leq p'$ and non-negative for all $p\geq p''$. Applying \eqref{eigen-bound} to $A$ here,
\begin{equation}
\begin{split}
\lambda_1(A) & \leq \lambda_1(p \rho_1) + \lambda_1(-(1-p)\rho_2) \\
& = p \lambda_1(\rho_1) - (1-p) \lambda_d(\rho_2) \\
& = p (\lambda_1(\rho_1) + \lambda_d(\rho_2)) - \lambda_d(\rho_2). \\
\end{split}
\end{equation}
Thus, $\lambda_1(A) \leq 0$ for all $p\leq p'$. The derivation of $p''$ works similarly.

The last statement follows from the fact that for every density matrix $\rho$, $0\leq\lambda_d(\rho)\leq \frac{1}{d} \leq \lambda_1(\rho)$, and therefore $\lambda_1(\rho_1) + \lambda_d(\rho_2)\geq 2\lambda_d(\rho_2)$ and $\lambda_d(\rho_1) + \lambda_1(\rho_2)\leq 2 \lambda_1(\rho_2)$.
\end{proof}

\subsection{Bounds on the Maximal Reliability}

In this subsection, we focus on bounds on $R^{\max}_p(\rho)$. A simple upper bound was already provided in \eqref{eqn:r-max-bounds-rho} of Cor.~\ref{cor:r-max-bounds-rho}. According to Prop.~\ref{prop:rhovsblind} and Prop.~\ref{rho-trivial-lower-bound}, $R^{\max}_p(\rho)= \max(p,1-p)$ when either $p\geq d/(d+1)$ or $p\leq p'$ as in \eqref{p-primedef}. Here is another upper bound.

It is convenient to express $\rho$ in terms of its spectral decomposition. Let $\phi_i$ for $i = 1, \ldots, d$ be an orthonormal basis of eigenvectors of $\rho$ with corresponding eigenvalues $p_i$, 
\begin{equation}\label{eqn:spectral-rho}
\rho = \sum_{i = 1}^d p_i \pr{\phi_i}.
\end{equation}

\begin{prop}
For any density matrix $\rho$ and $0 < p < d/(d+1)$,
\begin{equation}\label{eqn:optimal-rho-reliability-bound}
R^{\max}_p(\rho) \leq 
p\left(1 + \sum\limits_{i = 1}^d p_i f^{-1}_{\phi_i}\Bigl(\frac{p}{1-p}\Bigr) \right) .
\end{equation}
\end{prop}

\begin{proof}
Note that for a general $E$, we may utilize \eqref{eqn:spectral-rho} in the following way.
\begin{equation}
\begin{split}
R_{\rho, p}(E) &= \tr \left[ \rho \Bigl( p \diag E + (1-p)(I - E) \Bigr) \right] \\
&= \tr \left[ \sum_{i = 1}^d p_i \pr{\phi_i} \Bigl(p \diag E + (1-p)(I-E) \Bigr) \right] \\
&= \sum_{i = 1}^d p_i \tr \Bigl[ \pr{\phi_i} \Bigl(p \diag E + (1-p)(I-E) \Bigr) \Bigr] \\
&= \sum_{i = 1}^d p_i R_{\phi_i, p}(E) \\
& \leq \sum_{i = 1}^d p_i R^{\max}_p(\phi_i).
\end{split}
\end{equation}
This, combined with Thm.~\ref{thm:complete-information-optimality}, gives the result.
\end{proof}

\subsection{Uniform Distribution}
\label{sec:uniform}

A special case of random $\psi$ that deserves separate discussion is that of a uniform distribution $\mu$, with a corresponding density matrix of $\rho = I/d$. Note that $\diag \rho = \rho$. That is, the density matrix $\rho$ of the uncollapsed state vector coincides with the density matrix $\diag \rho$ of the collapsed state vector, so in terms of distinguishing between two density matrices, we would have to distinguish between two equal density matrices. It follows immediately that no experiment can detect whether a collapse has occurred. In fact, it follows that no experiment can yield any probabilistic information at all about whether a collapse has occurred (also if the set $\mathcal{Z}$ of possible outcomes has more than two elements), that is, the distribution of the outcome $Z$ satisfies
\be
\PPP(Z=z|\text{collapse}) = \PPP(Z=z|\text{no collapse})\,. 
\ee
So, Alice can do no better than blind guessing. Of course, it follows also that the reliability cannot exceed that of blind guessing,  $R^{\max}_p(\mu) = R_{\mu, p}(\E_\emptyset)$ with $\E_\emptyset$ = blind guessing. More precisely:

\begin{theorem}\label{thm:uniform}
For $\psi$ uniformly random from $\SSS$ and $p \neq 1/2$, any non-trivial experiment (i.e., one with $0\neq E\neq I$) is strictly less reliable than blind guessing. For $p = 1/2$, any non-trivial experiment is exactly as reliable as blind guessing. That is, for $\mu$ uniform on $\SSS$, $R_{\mu, p}(E) \leq R_{\mu, p}(\E_\emptyset)$ for all $p \in [0,1]$ and all $0 \leq E \leq I$, with equality only for $p = 1/2$, or $E=0$, or $E=I$.
\end{theorem}

\begin{proof}
For an arbitrary operator $0 \leq E \leq I$, we have that
\begin{equation}\label{eqn:uniform-reliability}
\begin{split}
R_{\rho, p}(E) & = \tr \left[ \rho (p \diag E + (1-p)(I-E) ) \right] \\
& = (1-p) + \tr \left[ E (p \diag \rho - (1-p) \rho) \right]  \\
& = (1-p) + (2p-1) \tr \left[ E \right]/d. \\
\end{split}
\end{equation}
From this, it is easy to see that if $p < 1/2$, reliability is maximized when $\tr E$ is minimized, or $E = 0$. If $p > 1/2$, reliability is maximized when $\tr E$ is maximized, or $E = I$. When $p = 1/2$, the reliability is in fact independent of $E$. 
\end{proof}

\subsection{Reduced Density Matrices and Other Scenarios}
\label{sec:rhored}

The fact that the reliability depends on $\mu$ only through $\rho$ suggests that the reliability has the same value for \emph{any} system with density matrix $\rho$, i.e., also for systems that have \emph{reduced} density matrix $\rho$. This is indeed the case, as we show in the first of the following four variations of our scenario:

\begin{enumerate}
\item Suppose that the system $S$ is entangled with another system $T$, and that Alice cannot do (or, at any rate, does not do) experiments on $T$, only on $S$. Any yes-no-experiment $\E$ on $S$ has a POVM of the form $\bigl(E\otimes I_T,(I_S-E)\otimes I_T\bigr)$. Suppose further that the composite system $ST$ has an initial state vector $\psi_{ST}$ that is random with distribution $\mu_{ST}$ on the unit sphere of $\Hilbert_{ST}=\Hilbert_S\otimes\Hilbert_T$. Suppose further that collapse, which occurs with probability $p$, affects only $S$, not $T$. 
That is, the state vector that Alice encounters is
\begin{equation}\label{eqn:psi-ST-prime-def}
\psi'_{ST} = \begin{cases}
\psi_{ST} & \text{with probability }1-p\\
b_k \otimes \frac{\scp{b_k}{\psi_{ST}}}{\|\scp{b_k}{\psi_{ST}}\|} & \text{with probability }p\,\bigl\| \scp{b_k}{\psi_{ST}} \bigr\|^2\text{ for }k=1,\ldots,d\,,
\end{cases}
\end{equation}
where $\{b_k:k=1,\ldots,d\}$ is an orthonormal basis of $\Hilbert_S$, the inner product $\scp{\cdot}{\cdot}$ is the partial inner product that yields a vector in $\Hilbert_T$, and probabilities are conditional on the given $\psi_{ST}$. 
It is then easy to verify that the reliability of $\E$ (i.e., the probability that $\E$ correctly retrodicts whether collapse has occurred) is
\be
R_{\mu_{ST},p}(\E)= R_{\rho,p}(E)
\ee
with $\rho$ the reduced density matrix obtained by a partial trace,
\be
\rho = \tr_T \int_{\SSS} \mu_{ST}(d\psi_{ST}) \pr{\psi_{ST}}\,.
\ee

\item 
Suppose now that the system $S$ is entangled with another system $T$, that Alice can do experiments only on $S$, and that collapse affects $T$, not $S$. That is, instead of a basis $\{b_k\}$ of $\Hilbert_S$, we are given a basis $\{\tilde b_j\}$ of $\Hilbert_T$, and the initial state vector $\psi_{ST}$ becomes
\be
\psi_{ST}' = \begin{cases} \psi_{ST} & \text{with probability }1-p\\
\frac{\braket{\tilde b_j|\psi_{ST}}}{\|\braket{\tilde b_j|\psi_{ST}}\|}\otimes \tilde b_j & \text{with probability } p \,\bigl\| \braket{\tilde b_j|\psi_{ST}} \bigr\|^2\,,\end{cases}
\ee
where $\braket{\cdot|\cdot}$ is the partial inner product in $\Hilbert_T$. Then, in terms of the problem of distinguishing between $\rho_1$ and $\rho_2$,
\be
\rho_1=\EEE_{\psi_{ST}\sim\mu_{ST}} \Bigl[ \tr_T \pr{\psi_{ST}} \Bigr]=\rho_2\,,
\ee
and no experiment is more reliable than blind guessing. In fact, no experiment can provide any information at all about whether collapse has occurred. (This fact is, of course, well know from the no-signaling theorems about EPR-type experiments, where experiments on one side $S$ of a bipartite entangled quantum system $ST$ cannot reveal information about whether any experiment was carried out on the other side $T$.)

\item Suppose again that $S$ is entangled with $T$, and that Alice has access only to $S$, but suppose now that collapse occurs, if it occurs, to a basis $\{\hat b_i\}$ of $\Hilbert_S\otimes\Hilbert_T$. That is, $\psi_{ST}$ becomes
\be
\psi_{ST}' = \begin{cases} \psi_{ST} & \text{with probability }1-p\\
\frac{\braket{\hat b_i|\psi_{ST}}}{|\braket{\hat b_i|\psi_{ST}}|}\hat b_i & \text{with probability } p \,\bigl| \braket{\hat b_i|\psi_{ST}} \bigr|^2\,,\end{cases}
\ee
where $\braket{\cdot|\cdot}$ is the inner product in $\Hilbert_S\otimes\Hilbert_T$. Then Helstrom's theorem applies with
\be
\begin{split}
\rho_1&=\EEE_{\psi_{ST}\sim\mu_{ST}} \Biggl[ \tr_T \sum_{i=1}^{d\,\dim\Hilbert_T} \pr{\hat b_i} \psi_{ST}\rangle\langle \psi_{ST} \pr{\hat b_i}\Biggr] \,,\\
\rho_2&=\EEE_{\psi_{ST}\sim\mu_{ST}} \Bigl[ \tr_T \pr{\psi_{ST}} \Bigr]\,.
\end{split}
\ee

\item Suppose now that there is no $T$ system, only $S$, and that collapse, if it occurs, does not project $\psi$ to 1-dimensional subspaces $\CCC b_k$ but to higher-dimensional subspaces $\Hilbert_k$, $k=1,\ldots,K<d$, where $\Hilbert=\oplus_k \Hilbert_k$ is the orthogonal sum. (This situation arises if collapse occurs by Bob performing a quantum measurement of a degenerate observable.) That is,
\be
\psi' = \begin{cases} \psi & \text{with probability }1-p\\
P_k\psi & \text{with probability } p \,\bigl\| P_k \psi \bigr\|^2\,,\end{cases}
\ee
where $P_k$ is the projection onto $\Hilbert_k$. Then Helstrom's theorem applies with
\be
\begin{split}
\rho_1&=\EEE_{\psi\sim\mu} \Bigl[ \sum_{k=1}^{K} P_k \pr{\psi} P_k \Bigr]\,,\\
\rho_2&=\EEE_{\psi\sim\mu} \bigl[\pr{\psi} \bigr]\,.
\end{split}
\ee
One could also relax the condition that the $P_k$ are projection operators and require only that $P_k\geq 0$ and $\sum_k P_k^2=I$, a kind of unsharp collapse. (Strictly speaking, this kind of collapse occurs in GRW theory.)
\end{enumerate}

\section{The Case Without Prior Information About $\psi$}
\label{sec:no-information-experimentation}

So far we assumed that the initial wave function $\psi$ is either known or randomly drawn from a known distribution $\mu$. Can one detect whether a wave function has collapsed, if no such information is given? We discuss this question in detail elsewhere \cite{CT12c} and report here the results.

The question can be thought of in the following way. Were Alice presented with an ensemble of systems that Bob may or may not have tampered with and caused to collapse, she could perform a sequence of experiments over multiple systems that would give her information about the distribution of the initial $\psi$, leading to the situations described in the previous two sections. However, if she is presented with only one such system, and told nothing about it, she cannot make any reasonable assumptions about how that initial $\psi$ was chosen.

There is one thing she can be sure of, however. Blind guessing provides a reliability of $\max(1-p, p)$ independent of the initial state or distribution of system $S$. Blind guessing is always feasible, and always provides that reliability, even if Alice has no prior knowledge of the system. The question is, can she do better? There are essentially two ways of considering this situation. 

In the first, Alice might consider taking the initial wave function $\psi$ as uniformly likely to be chosen anywhere on $\SSS$. This corresponds to the Bayesian notion of having no information about the initial state of $S$---a uniform distribution $\mu$. We have discussed this distribution in Section~\ref{sec:uniform} above, and the upshot is that, from this Bayesian perspective, Alice can do no better than blind guessing.

However, Alice might be unwilling to make the assumption that $\psi$ is uniformly distributed---with no prior information, how could she justify this assumption? Taking an assumption-free approach to model Alice's lack of information, we can instead ask: For a particular experiment $\E$, for what fraction of $\SSS$ does $\E$ perform better than blind guessing? Thm.~\ref{thm:uniform} demonstrates that, averaged over the entire sphere, no experiment is more reliable than blind guessing. However, Alice does not need an experiment that performs well over the whole sphere---merely one that performs well for the system she is presented with. If the fraction of $\SSS$ for which $\E$ performs well is large, at least $1/2$ the sphere, Alice may feel comfortable using $\E$ instead of blind guessing. However, there are limits to this strategy \cite{CT12c}:

To begin with, for any experiment $\E$ and $p\neq \frac12$, the set of $\psi\in\SSS$ where $\E$ is more reliable than blind guessing has less than full measure. Let us write $\Lambda_p(\E)$ for the normalized measure of that set (i.e., for the fraction of the sphere where $\E$ is more reliable than blind guessing); in fact, $\Lambda_p(\E)$ depends on $\E$ only through $E$, $\Lambda_p(\E)=\Lambda_p(E)$. A key question is whether $\Lambda_p(E)\leq \frac12$ or $\Lambda(E)>\frac12$. In the former case, it seems that no experiment is more useful than blind guessing. We find that this case occurs in dimension $d=2$ for any $0\leq E\leq I$ and $0<p<1$, as well as when $p < 1/2 - 1/\sqrt{8} \approx 0.146$ or $p > 1/2 + 1/\sqrt{8} \approx 0.854$ for any $d$ and any $0\leq E\leq I$. We have also found further, more complicated, sufficient conditions for $\Lambda_p(E)\leq \frac12$. However, we have also found that for every $d\geq 3$, for some values of $p$, there exist operators $0\leq E\leq I$ such that $\Lambda_p(E)>\frac12$. Moreover, we have found reason to conjecture that for all $d\geq 2$, all $0 \leq E \leq I$, and all $0 < p < 1$,
\begin{equation}
\label{eqn:that-damn-conjecture}
\Lambda_p(E) \leq 1 - \left( 1 - \frac{1}{d}  \right)^{d-1}\,.
\end{equation}
In particular, $\Lambda_p(E) \leq 1 - 1/e \approx 0.632$. Thus, some experiments may be more reliable than blind guessing for more than 50\%, but apparently not for more than 64\%\ of the sphere.

\bigskip

\textit{Acknowledgments.}
Both authors are supported in part by NSF Grant SES-0957568. 
R.T.~is supported in part by grant no.~37433 from the John Templeton Foundation and by the Trustees Research Fellowship Program at Rutgers, the State University of New Jersey. 

\end{document}